\documentclass[nolimit]{article}
\usepackage{amsmath,amssymb,amsthm,mathrsfs,amssymb,mathrsfs}
\usepackage{graphicx}
\newtheorem{thm}{Theorem}[section]

\newtheorem{lem}[thm]{Lemma}

\newtheorem{example}{Example}[section]
\theoremstyle{definition}

\newtheorem{theorem}[thm]{Theorem}

\newcommand{\bra}[1]{\mbox{$\langle #1 |$}}
\newcommand{\ket}[1]{\mbox{$| #1 \rangle$}} 
 
\newcommand{\tr}{\mathbf{Tr}}

\begin{document}
\date{}
\title{\bf  Monogamy Relations for the Generalized W Class}
\author{Xian Shi$^{1,2,3}$\\
	{\footnotesize $^1$Institute of Mathematics, Academy of Mathematics and Systems Science,}\\ {\footnotesize Chinese Academy of Sciences, Beijing 100190, China}\\
	{\footnotesize $^2$	University of Chinese Academy of Sciences, Beijing 100049, China,}\\
	{\footnotesize $^3$	UTS-AMSS Joint Research Laboratory for Quantum Computation and Quantum Information Processing,}\\
	{\footnotesize Academy of Mathematics and Systems Science, Chinese Academy of Sciences, Beijing 100190, China}}
\maketitle	
\begin{abstract}
\indent Here we investigate the monogamy relations for the generalized W class. Monogamy inequalities for the generalized W class in terms of the $\beta$-th ($\beta\in (0,2)$) power of the concurrence, the concurrence of assistance and the negativity are presented. At last, under some restricted conditions of the generalized W class, we present stronger monogamy inequalities.
\end{abstract}

\section{Introduction}
\indent Monogamy of entanglement is an interesting property that characterizes the distribution
of entanglement, it presents that entanglement cannot be shareable arbitrarily among many
parties, which is different from classical correlations \cite{T}. If party A has strong correlation
with party B such that $|{\psi}\rangle_{AB}=\frac{|{00}\rangle+|{11}\rangle}{\sqrt{2}}$, then the correlations between A and B
cannot be shared by party C, that is, $\rho_{ABC}=\rho_{AB}\otimes\rho_C.$ This property has been
applied on many tasks in quantum information, for example, it can be applied on the
proof of the security of quantum cryptography \cite{JM}.\\
\indent Mathematically, for a tripartite system A, B and C, the general monogamy in terms of
an entanglement measure $\mathcal{E}$ implies that the entanglement between A and BC satisfies
\begin{align}
\mathcal{E}_{A|BC}\ge \mathcal{E}_{A|B}+\mathcal{E}_{A|C}.
\end{align}
This relation was first proved for qubit systems with respect to the squared concurrence
\cite{CKW,OV}. Moreover, it has been shown that the monogamy relation is valid for the $\alpha$-th power of concurrence for the qubit systems. Negativity is another useful entanglement measure. It has been showed that the monogamy relations is valid for the three-qubit systems in terms of the negativity \cite{OF,HG}. Similarly, this relation is also generalized to the $\alpha$-th power of the negativity for n-qubit systems \cite{LL}. Recently, tighter monogamy inequalities were presented for qubit systems \cite{JF,ZF}. As a dual of concurrence, concurrence of assistance (CoA)is also meaningful, as the first polygamy inequality is shown to be valid for n-qubit pure states in terms CoA \cite{GBS}. Recently, Luo present an index, monogamy power, to depict the monogamy relation in terms of an entanglement measure $E$ \cite{L}.\\
\indent However, the above relations are not valid for higher dimensional systems , there exists counterexamples for states in the systems $3\otimes 3\otimes 3$ \cite{O} and $3\otimes 2\otimes 2$ \cite{KDS}. Although the monogamy inequality is invalid for higher dimensional systems, there exists a class of n-qudit pure states, the generalized W class, satisfy the monogamy inequalities in terms of the squared concurrence \cite{JS} and the squared convex roof extended negativity (CREN) \cite{KDS,JK}, moreover, this class saturates the inequality.\\
\indent This article is organized as follows, first we review some preliminary knowledge needed. Then we present the monogamy inequality in terms of the $\beta$-th power of CoA and CREN for the tripartite generalized W class when $\beta>0$, under some restricted conditions we make, we present a monogamy inequality in terms of the $\beta$-th power of CoA and CREN for multipartite generalized W class when $\beta>0.$

\section{Preliminaries}
For a bipartite pure state $|\psi\rangle_{AB}$,
the concurrence is given by 
\begin{equation}\label{CON}
C(|\psi\rangle_{AB})=\sqrt{2[1-Tr(\rho^2_A)]},
\end{equation}
where $\rho_A$ is the reduced density matrix by tracing over the subsystem $B$,
$\rho_{A}=Tr_{B}(|\psi\rangle_{AB}\langle\psi|)$.
The concurrence is extended to mixed states $\rho$
 by the convex roof construction
\begin{align} 
C(\rho_{AB})=\min_{\{p_i,|\psi_i\rangle\}} \sum_i p_i C(|\psi_i\rangle).
\end{align}
where the minimization takes over all the decompositions $\{p_i,|\psi_i\rangle\}$ of $\rho=\sum_{i}p_{i}|\psi _{i}\rangle \langle \psi _{i}|$ with $p_{i}\geq 0$, $\sum_{i}p_{i}=1$.
For an n-qubit quantum states, when $\alpha\ge 2,$ the concurrence satisfies \cite{ZNF}
\begin{align} \label{a1}
C^{\alpha}_{A|B_1B_2...B_{n-1}}\geq C^{\alpha}_{AB_1}+...+C^{\alpha}_{AB_{n-1}},
\end{align}
where we assume $C_{AB_i}$, $i=1,2...,n-1$, is the
concurrence of the mixed states $\rho_{AB_i}=Tr_{B_1B_2...B_{i-1}B_{i+1}...B_{n-1}}(\rho)$.
If $C_{AB_i}\not=0$, $i=1,...,n-1$, the concurrence satisfies
\begin{equation}\label{a2}
C^{\alpha}_{A|B_{1}...B_{n-1}}
\le C^{\alpha}_{AB_{1}}+...+C^{\alpha}_{AB_{n-1}},
\end{equation}
for $\alpha\leq0$, where we assume $C_{AB_i}$, $i=1,2...,n-1$, is the
concurrence of the mixed states $\rho_{AB_i}=Tr_{B_1B_2...B_{i-1}B_{i+1}...B_{n-1}}(\rho)$. Furthermore, in \cite{JF, ZF} tighter monogamy inequalities than (\ref{a1}) are derived for the $\alpha$th $(\alpha\geq2)$ power of concurrence.
\\
\indent Dual to the concurrence, for the mixed state $\rho$, the CoA $C^a$ is defined as 
\begin{align}
C^a(\rho_{AB})=\max_{\{p_i,|\psi_i\rangle\}} \sum_i p_i C(|\psi_i\rangle).
\end{align}
where the maximum takes over all the decompositions $\{p_i,|\psi_i\rangle\}$ of $\rho=\sum_{i}p_{i}|\psi _{i}\rangle \langle \psi _{i}|$ with $p_{i}\geq 0$, $\sum_{i}p_{i}=1$.
In \cite{GBS}, the author showed that for an n-qubit pure state $|{\psi}\rangle_{A|B_1B_2...B_{n-1}}$, there exists such polygamy inequalities
\begin{align}
(C^a(|{\psi}_{A|B_1B_2...B_{n-1}}))^{2}\leq (C^a(\rho_{AB_1}))^{2}+...+(C^a(\rho_{AB_{n-1}}))^{2},
\end{align}
where we assume $C_{AB_i}$, $i=1,2...,n-1$, is the
concurrence of the mixed states $\rho_{AB_i}=Tr_{B_1B_2...B_{i-1}B_{i+1}...B_{n-1}}(\rho)$.\\

At last, let us recall the definition of the generalized W class states \ket{W} \cite{JS},
\begin{align}
	\ket{\psi}=\sum_{i=1}^da_{1i}\ket{i00\cdots 0}+a_{2i}\ket{0i0\cdots 0}+\cdots+a_{ni}\ket{000\cdots i},
\end{align}here we assume $\sum_{ij}|a_{ij}|^2=1.$
There the authors present the following lemma \cite{JS}.
\begin{lem}
	For any n-qudit generalized W-class state $\ket{\psi}_{AB_{1}\cdots B_{n-1},}$ in (8) and a partition $P=\{P_1,\cdots,P_m\}$ for the set of subsystems $S=\{A, B_{1}, \cdots, B_{n-1}\},$
	\begin{align}
	C^2_{P_1\cdots \overline{P_s}\cdots P_m}=\sum_{k\ne s}C^2_{P_sP_k}=\sum_{k\ne s}(C^a_{P_sP_k})^2,
	\end{align}
	and \begin{align}
	C_{P_sP_k}=C^a_{P_sP_k},
	\end{align}
	for all $k\ne s$.
\end{lem}
\section{Monogamy Inequalities for CoA  }
\indent First we will present a lemma, this lemma is useful for the results below.
\begin{lem}
For real numbers $x\in[0,1]$  and $t\geq1$, we have 
\begin{align}\label{lem}
(1+t)^x\geq1+(2^x-1)t^x.
\end{align}
\end{lem}
\begin{proof}
The inequality (\ref{lem}) can be seen as a question to find the biggest value of the function $g_{x}(t)=\frac{(1+t)^x-1}{t^x}$ about t.
Then we have when $t\ge 1,$ $\frac{d g_{x}(t)}{d t}=xt^{-(x+1)}[1-(1+t)^{x-1}]\geq0$, that is, $g_{x}(t)$ is an
increasing function of $t$. Hence, when $x\in [0,1]$, $g_{x}(t)\geq g_{x}(1)$, i.e, $(1+t)^x\geq1+(2^x-1)t^x.$
\end{proof}
\indent When we denote the partition $\{P_1,P_2,P_3\}$ is a subset of the set $\{A,B_1,B_2,\cdots,B_{n-1}\}$, the equality (9) and (10) become
\begin{align}
C^2_{P_1|P_2P_3}=C^2_{P_1|P_2}+C^2_{P_1|P_3}\\
C_{P_1|P_2}=C_{P_1|P_2}^a.
\end{align} 
The equality (12) can be extended to 
\begin{align}
 C^{\alpha}_{P_1|P_2P_3}\ge C^{\alpha}_{P_1|P_2}+C^{\alpha}_{P_1|P_3},
\end{align} 
as 
\begin{align}
C^{\alpha}_{P_1|P_2P_3}= &(C^2_{P_1|P_2}+C^2_{P_1|P_3})^{\alpha/2}\nonumber\\
= & C^{\alpha}_{P_1|P_2}(1+\frac{C^2_{P_1|P_3}}{C^2_{P_1|P_2}})^{\alpha/2}\nonumber\\
\ge&C^{\alpha}_{P_1|P_2}+C^{\alpha}_{P_1|P_3}
\end{align}
when $\alpha\ge 2.$ The inequality in $(15)$ is due to $(1+\frac{C^2_{P_1|P_3}}{C^2_{P_1|P_2}})^{\alpha/2}\ge 1+\frac{C^2_{P_1|P_3}}{C^2_{P_1|P_2}}.$
Next we present the following theorem.
\begin{theorem}
	Assume $\ket{\psi}_{AB_1B_2\cdots B_{n-1}}$ is a generalized W class state, when we denote the partition $\{P_1,P_2,P_3\}$ is a subset of the set $\{A,B_1,B_2,\cdots,B_{n-1}\}$, then we have the following inequalities,
	\begin{align}
	(C^a_{P_1|P_2P_3})^{\beta}\ge (2^{\frac{\beta}{\alpha}}-1)\max{(C^a_{P_1|P_2})^{\beta},(C^a_{P_1|P_3})^{\beta}}) +\min({(C^a_{P_1|P_2})^{\beta},(C^a_{P_1|P_3})^{\beta}}), 
	\end{align}
	when $0\le \beta\le \alpha$ and $\alpha\ge 2.$
\end{theorem}
\begin{proof}
		 If $C^a_{P_1|P_2}\le C^a_{P_1|P_3},$ and $C^a_{P_1|P_2}= 0,$ the inequality is trivial, as $2^{\frac{\beta}{\alpha}}-1\le 1$ and $(C^a_{P_1|P_2P_3})\ge (C^a_{P_1|P_3})$.
		 \\
	 If $C_{P_1|P_2}\le C_{P_1|P_3},$ and $C_{P_1|P_2}\ne 0,$
we have that 
	\begin{align}
(C^a_{P_1|P_2P_3})^{\beta}=&(C_{P_1|P_2P_3})^{\beta}\nonumber\\
\ge &(C_{P_1|P_2}^{\frac{\beta}{\alpha}})^{\alpha}+(C_{P_1|P_3}^{\frac{\beta}{\alpha}})^{\alpha}\nonumber\\
= &C_{P_1|P_2}^{{\beta}}*(\frac{C_{P_1|P_3}^{\alpha}}{C_{P_1|P_2}^{\alpha}}+1)^{\frac{\beta}{\alpha}}\nonumber\\
\ge& C_{P_1|P_2}^{{\beta}}*(1+(2^{\frac{\beta}{\alpha}}-1)(\frac{C_{P_1|P_3}^{\alpha}}{C_{P_1|P_2}^{\alpha}})^{\frac{\beta}{\alpha}})\nonumber\\
=&C_{P_1|P_2}^{{\beta}}+(2^{\frac{\beta}{\alpha}}-1)C_{P_1|P_3}^{{\beta}}\nonumber\\
=&(C^a_{P_1|P_2})^{{\beta}}+(2^{\frac{\beta}{\alpha}}-1)(C^a_{P_1|P_3})^{{\beta}}
	\end{align}
here the first inequality is due to the inequality (14), the second inequality is due to the lemma 3.1.\\
\indent	Similarly, we can get the case when $C_{P_1|P_2}\ge C_{P_1|P_3}.$
\end{proof}
\indent Furthermore, the theorem 3.2 tells us a general monogamy inequality for the W class states in terms of CoA. And we have that the existing result (15) can be seen as a corollary of the inequality (16).\\
\indent Next we present an example to present our theorem.
\begin{example}
	Here let us take a 3-qubit generalized W state, $\ket{\psi}=\frac{1}{\sqrt{6}}(\ket{001}+\ket{010}+2\ket{100})$, through computation, we have $C_{A|B}=1/3,$ $C_{A|C}=2/3.$ then the equality (18) is 
	\begin{align}
	f(\beta,\alpha)=	&(C^a_{P_1|P_2P_3})^{\beta}- [(2^{\frac{\beta}{\alpha}}-1)\max({(C^a_{P_1|P_2})^{\beta},(C^a_{P_1|P_3})^{\beta}}) +\min({(C^a_{P_1|P_2})^{\beta},(C^a_{P_1|P_3})^{\beta}})]\nonumber\\
		=&(\frac{\sqrt{5}}{3})^{\beta}-(2^{\frac{\beta}{\alpha}}-1)*(\frac{2}{3})^{\beta}-(\frac{1}{3})^{\beta}\ge 0
	\end{align}
	From (20), when $\alpha=2,$ the function $f(\beta,2)\ge f(\beta,\alpha)$, $\alpha>2.$ And when $\alpha=2,$ through rough computation, the function $f(\beta,2)$ is decreasing. As the figure tells us, the formula is bigger than 0.
\end{example}
	\begin{figure}
	\centering
	\includegraphics[width=100mm]{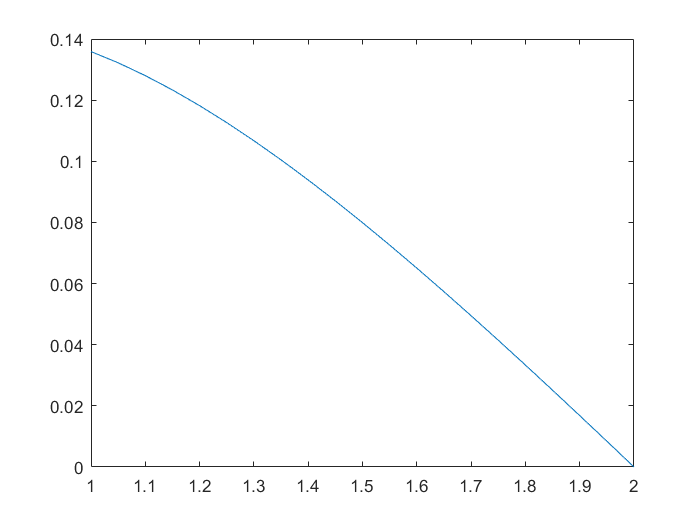}\\
	\caption{In this figure, we plot the function of $f(\beta,2)$. As when $\alpha=2$, the function $f(\beta)_{\alpha}$ is optimal. }\label{}
\end{figure}
\indent We can generalize the theorem 3.2 to the multipartite generalized W class states.\\
\begin{theorem}
Let $\rho_{P_1,\cdots,P_m}$ is a reduced density matrix of a generalized W class state $\ket{\psi}_{AB_1\cdots B_{n-1}},$ when $C^a_{P_1P_i}\le C^a_{P_1P_{i+1}\cdots P_{m-1}}, i=1,2,\cdots,n-1$ and $C^a_{P_1P_{j}}\ge C^a_{P_1|P_{j+1}\cdots P_{m}},j=n,\cdots,m-1,$ then we have 
\begin{align}
(C^a_{P_1|P_2\cdots P_m})^{\beta}\ge& \sum_{i=2}^nh^{i-1}(C^a_{P_1|P_i})^{\beta}+h^{n}\sum_{i=n+1}^{m-1}(C^a_{P_1|P_i})^{\beta}+h^m(C^a_{P_1|P_{m}})^{\beta}.
\end{align}
where $\beta\in[0,\alpha],$ $\alpha\ge 2$. And we denote that $h=2^{\frac{\beta}{\alpha}}-1.$
\end{theorem}
\begin{proof}
	From the theorem 3.2, we have that 
	\begin{align}
(C^a_{P_1|P_2\cdots P_m})^{\beta}\nonumber\\
	\ge &(C^a_{P_1|P_2})^{\beta}+h (C^{a}_{P_1|P_3\cdots P_{m}})^{\beta} \\
	\ge & \cdots \\
	\ge & \sum_{i=2}^n h^{i-2} (C^a_{P_1|P_i})^{\beta}+ h^{n-1} (C^{a}_{P_1|P_{n+1}\cdots P_{m}})^{\beta} \\
	\ge &\sum_{i=2}^n h^{i-2} (C^a_{P_1|P_i})^{\beta}+ h^{n-1}(h(C^{a}_{P_1|P_{n+1}})^{\beta}+(C^a_{P_1|P_{n+2}\cdots P_{m}})^{\beta}) \\
	\cdots \\
\ge &\sum_{i=2}^{n}h^{i-2} (C^{a}_{P_1|P_i})^{\beta}+h^{n}\sum_{i=n+1}^{m-1}(C^a_{P_1|P_{i}})^{\beta}+h^m(C^a_{P_1|P_{m} })^{\beta}.\nonumber\\
=&\sum_{i=2}^nh^{i-1}(C^a_{P_1|P_i})^{\beta}+h^{n}\sum_{i=n+1}^{m-1}(C^a_{P_1|P_i})^a+h^m(C^a_{P_1|P_{m}})^{\beta}.
	\end{align}
	where the inequalities (20)-(22) is due to $C_{P_1P_i}\le C_{P_1P_{i+1}\cdot P_{m-1}}, i=1,2,\cdots,n-1$ , and the inequalities (23)-(25) is due to $C_{P_1P_{j}}\ge C_{P_1|P_{j+1}\cdots P_{m}},j=n,\cdots,m-1.$
\end{proof}
\indent Similar to the analysis of the theorem 3.2, we present all the powers of the generalized W class in terms of CoA under some restricted conditions.
\section{monogamy inequalities for the negativity}
\indent In this section, first let us recall the definition of negativity. For a bipartite pure state $\ket{\psi}_{AB}=\sum_{i=0}^{d-1} \sqrt{\lambda_i}\ket{ii},$ its negativity is defined as 
\begin{align}
N(\ket{\psi}_{A|B})=||\ket{\psi}_{A|B}\bra{\psi}||_1-1=\sum_{i<j}2\sqrt{\lambda_i\lambda_j},
\end{align}
where $$\ket{\psi}\bra{\psi}^{T_B}=\sum_{i,j=0}^{d-1}\sqrt{\lambda_i\lambda_j}\ket{ij}_{AB}\bra{ji}$$ is the partial trace of $\ket{\psi}_{AB}$, and $||\cdot||_1$ is the 1-norm.\\
\indent For a mixed state $\rho_{AB},$ here we denote its negativity as CREN by the convex roof extended method as  \cite{LOK}
\begin{align}
\mathcal{N}(\rho_{AB})=\min_{\{p_i,|\psi_i\rangle\}} \sum_i p_i N(|\psi_i\rangle).
\end{align}
where the minimization takes over all the decompositions $\{p_i,|\psi_i\rangle\}$ of $\rho=\sum_{i}p_{i}|\psi _{i}\rangle \langle \psi _{i}|$ with $p_{i}\geq 0$, $\sum_{i}p_{i}=1$.\\
\indent In \cite{KDS, JK}, the authors showed that the generalized W class states $\ket{\psi}_{AB_1B_2...B_{n-1}}$ satisfy the monogamy relations in terms of the negativity:
\begin{align}
\mathcal{N}^{2}_{A|B_1B_2...B_{n-1}}= \mathcal{N}^{2}_{AB_1}+...+\mathcal{N}^{2}_{AB_{n-1}},
\end{align}
where we assume that $\rho_{AB_i}=\tr_{B_1B_2\cdots B_{i-1}B_{i+1}\cdots B_{n-1}}\ket{\psi}\bra{\psi},$ $i=1,2,\cdots,n-1.$ Similar to the analysis of the inequality (15), the inequality (28) can be generalized in terms of the $\alpha$-th power of the CREN,
\begin{align}
\mathcal{N}^{\alpha}_{A|B_1B_2\cdots B_{n-1}}\ge \mathcal{N}^{\alpha}_{A|B_1}+\cdots+\mathcal{N}^{\alpha}_{A|B_{n-1}},
\end{align} 
where we assume that $\alpha\ge 2.$\\
\indent By the similar proof of the theorem 3.2 and theorem 3.4, we have the following theorems.
\begin{theorem}
	Assume $\ket{\psi}_{AB_1B_2}$ is a generalized W class state, then we have the following inequalities,
	\begin{align}
	(\mathcal{N}_{A|B_1B_2})^{\beta}\ge (2^{\frac{\beta}{\alpha}}-1)\max({(\mathcal{N}_{A|B_1})^{\beta},(\mathcal{N}_{A|B_2})^{\beta}}) +\min({(\mathcal{N}_{A|B_1})^{\beta},(\mathcal{N}_{A|B_2})^{\beta}}), 
	\end{align}
	when $0\le \beta\le \alpha$ and $\alpha\ge 2.$ 
\end{theorem}
\begin{theorem}
	Assume $\ket{\psi}_{AB_1\cdots B_{n-1}}$ is a generalized W class state, when $\mathcal{N}_{AB_i}\le \mathcal{N}_{AB_{i+1}\cdots B_{n-1}},$ $i=1,2,\cdots,m$ and $\mathcal{N}_{AB_{j}}\ge \mathcal{N}_{A|B_{n+1}\cdots B_{n-1}},j=m,\cdots,n-2,$ then we have 
	\begin{align}
	\mathcal{N}^{\beta}_{A|B_1\cdots B_{n-1}}\ge& \sum_{i=1}^mh^{i-1}\mathcal{N}^{\beta}_{A|B_i}+h^{m+1}\sum_{i=m+1}^{n-2}\mathcal{N}^{\beta}_{A|B_i}+h^m\mathcal{N}^{\beta}_{A|B_{n-1}}
	\end{align}
	where $\beta\in[0,\alpha],$ $\alpha\ge 2$. And we denote that $h=2^{\frac{\beta}{\alpha}}-1.$
\end{theorem}
\section{Conclusion}
\indent Monogamy of entanglement is a fundamental property of multipartite entanglement theory. Due to the importance on the study of the generalized W class states, here we mainly present the monogamy inequalities in terms of the $\beta$-th power of CoA and the $\beta$-th power of CREN ($\beta\in (0,2)$) for the generalized W class. Due to the importance of the study on monogamy of entanglement, our result can provide a rich reference on the study of multipartite entanglement theory for future work.

\end{document}